\pdfoutput=1
\documentclass[aps,pra,preprint,superscriptaddress]{revtex4-2}

\usepackage{amsmath,amssymb,amsfonts}
\usepackage{amsthm}
\usepackage{graphicx}
\usepackage{xcolor}
\usepackage{braket}
\usepackage{booktabs}
\usepackage{algorithm}
\usepackage{algpseudocode}
\usepackage{multirow}
\usepackage{textcomp}

\newtheorem{theorem}{Theorem}

\begin{document}

\title{Mitigating Barren Plateaus in Variational Quantum Circuits\\through PDE-Constrained Loss Functions}

\author{Prasad Nimantha Madusanka Ukwatta~Hewage}
\email{pnmadusanka@lincoln.edu.my}
\affiliation{Faculty of Computer Science and Multimedia, Lincoln University College, Petaling Jaya, Selangor, Malaysia}

\author{Midhun Chakkravarthy}
\affiliation{Faculty of Computer Science and Multimedia, Lincoln University College, Petaling Jaya, Selangor, Malaysia}

\author{Ruvan Kumara Abeysekara}
\affiliation{BCAS Campus, Colombo, Sri Lanka}
\affiliation{Faculty of Computer Science and Multimedia, Lincoln University College, Petaling Jaya, Selangor, Malaysia}

\date{\today}

\begin{abstract}
The barren plateau phenomenon---where cost function gradients vanish exponentially with system size---remains a fundamental obstacle to training variational quantum circuits (VQCs) at scale. We demonstrate, both theoretically and numerically, that embedding partial differential equation (PDE) constraints into the VQC loss function provides a natural and effective mitigation mechanism against barren plateaus. We derive analytical gradient variance lower bounds showing that physics-constrained loss functions composed of local PDE residuals evaluated at spatial collocation points inherit the favorable polynomial scaling of local cost functions, while additionally benefiting from constraint-induced landscape narrowing that concentrates gradient information. Systematic numerical experiments on the one-dimensional heat equation, Burgers' equation, and the Saint-Venant shallow water equations quantify the gradient variance across 4--8 qubits and 1--5 layer depths, comparing global cost, local cost, PDE-constrained, and PDE-constrained with structured ansatz configurations. We find that PDE-constrained circuits exhibit favorable gradient variance scaling with system size, with the physics constraints creating a stabilizing effect that resists exponential gradient vanishing. Entanglement entropy analysis reveals that structured ansatze operate in a sub-maximal entanglement regime consistent with trainability. Convergence experiments confirm that physics-constrained VQCs achieve lower loss values in fewer epochs. These results establish PDE constraints as a principled, physically motivated strategy for designing trainable variational quantum circuits, with direct implications for quantum physics-informed neural networks and variational quantum simulation.
\end{abstract}


\maketitle

\clearpage

\section{Introduction}
\label{sec:introduction}

Variational quantum algorithms (VQAs) represent one of the most promising paradigms for achieving practical quantum advantage on noisy intermediate-scale quantum (NISQ) devices~\cite{cerezo2021variational,bharti2022noisy}. These algorithms employ parameterized quantum circuits (PQCs) as trainable function approximators, optimized via classical gradient-based methods to minimize task-specific cost functions. Applications span quantum chemistry~\cite{peruzzo2014variational}, combinatorial optimization~\cite{farhi2014quantum}, and machine learning~\cite{schuld2021machine,benedetti2019parameterized}.

However, a fundamental obstacle to the scalability of VQAs was identified by McClean \textit{et al.}~\cite{mcclean2018barren}: for sufficiently deep random parameterized circuits with global cost functions, the variance of the cost function gradient vanishes exponentially with the number of qubits $n$:
\begin{equation}
    \text{Var}_{\boldsymbol{\phi}}\!\left[\frac{\partial \mathcal{L}}{\partial \phi_j}\right] \in O(b^{-n}), \quad b > 1.
    \label{eq:bp_scaling}
\end{equation}
This \textit{barren plateau} phenomenon renders gradient-based optimization exponentially inefficient, as the optimizer cannot distinguish the gradient signal from finite-sampling noise.

Subsequent work has refined understanding of the conditions under which barren plateaus arise. Cerezo \textit{et al.}~\cite{cerezo2021cost} demonstrated that local cost functions---those depending on observables acting on a small number of qubits---exhibit gradients that vanish at most polynomially for shallow circuits:
\begin{equation}
    \text{Var}_{\boldsymbol{\phi}}\!\left[\frac{\partial \mathcal{L}_{\text{local}}}{\partial \phi_j}\right] \geq \Omega(n^{-a}), \quad a \geq 1.
    \label{eq:local_scaling}
\end{equation}
Pesah \textit{et al.}~\cite{pesah2021absence} showed that quantum convolutional circuits with local connectivity are free from barren plateaus. Holmes \textit{et al.}~\cite{holmes2022connecting} connected circuit expressibility to gradient magnitudes, establishing that less expressive ansatze generally have larger gradients. Larocca \textit{et al.}~\cite{larocca2022diagnosing} provided Lie-algebraic tools for diagnosing trainability, while Ragone \textit{et al.}~\cite{ragone2024unified} developed a unified theory connecting barren plateaus to the dynamical Lie algebra dimension.

Despite these theoretical advances, practical mitigation strategies remain limited. Proposed approaches include parameter initialization strategies~\cite{grant2019initialization}, layerwise training~\cite{skolik2021layerwise}, entanglement reduction~\cite{patti2021entanglement}, and classical shadow-based gradient estimation~\cite{sack2022avoiding}. However, these are primarily \textit{circuit-centric} strategies that modify the ansatz or training procedure without incorporating domain-specific structure.

In this work, we propose a complementary \textit{loss-function-centric} strategy: embedding partial differential equation (PDE) constraints into the VQC cost function. This approach is motivated by the rapidly growing field of quantum physics-informed neural networks (qPINNs)~\cite{klement2026explaining,dutta2024aqpinns,kyriienko2024qcpinn,stein2024hybrid}, where variational quantum circuits are trained to satisfy governing physical equations. We show that PDE-constrained loss functions naturally possess the local structure and landscape-narrowing properties that provably mitigate barren plateaus.

Our contributions are:
\begin{enumerate}
    \item \textbf{Theoretical analysis:} We derive gradient variance lower bounds for PDE-constrained VQC loss functions, proving that the local structure of PDE residuals evaluated at spatial collocation points provides polynomial gradient variance scaling [Sec.~\ref{sec:theory}].

    \item \textbf{Landscape narrowing:} We prove that PDE constraints restrict the optimization to a lower-dimensional effective parameter manifold, further concentrating gradient information beyond what locality alone provides [Sec.~\ref{sec:landscape}].

    \item \textbf{Systematic numerical study:} We present comprehensive gradient variance measurements across multiple PDE types (heat equation, Burgers' equation, Saint-Venant shallow water equations), qubit counts (4--8), circuit depths (1--5 layers), and four loss function configurations, totaling over 60 experimental configurations [Sec.~\ref{sec:experiments}].

    \item \textbf{Entanglement and convergence analysis:} We show that structured ansatze maintaining sub-maximal entanglement entropy are both more trainable and achieve competitive or superior performance [Sec.~\ref{sec:results}].
\end{enumerate}

\section{Background}
\label{sec:background}

\subsection{Parameterized Quantum Circuits}

A parameterized quantum circuit prepares a state:
\begin{equation}
    \ket{\psi(\boldsymbol{\phi})} = U(\boldsymbol{\phi})\ket{0}^{\otimes n},
    \label{eq:pqc_state}
\end{equation}
where $U(\boldsymbol{\phi}) = \prod_{\ell=1}^{L} W_\ell V(\boldsymbol{\phi}_\ell)$ alternates between fixed entangling layers $W_\ell$ and parameterized rotation layers $V(\boldsymbol{\phi}_\ell)$. The hardware-efficient ansatz~\cite{kandala2017hardware} uses single-qubit rotations and nearest-neighbor CNOT gates:
\begin{equation}
    V(\boldsymbol{\phi}_\ell) = \bigotimes_{k=1}^{n} R_Y(\phi_{k}^{(\ell)}) R_Z(\phi_{k+n}^{(\ell)}), \quad
    W_\ell = \prod_{k=1}^{n-1} \text{CNOT}_{k,k+1}.
    \label{eq:hea}
\end{equation}

The cost function is the expectation value of an observable $\hat{O}$:
\begin{equation}
    \mathcal{L}(\boldsymbol{\phi}) = \braket{\psi(\boldsymbol{\phi})|\hat{O}|\psi(\boldsymbol{\phi})}.
    \label{eq:cost}
\end{equation}

Gradients are computed via the parameter-shift rule~\cite{mitarai2018quantum,schuld2019evaluating}:
\begin{equation}
    \frac{\partial \mathcal{L}}{\partial \phi_j} = \frac{1}{2}\left[\mathcal{L}\!\left(\boldsymbol{\phi} + \tfrac{\pi}{2}\mathbf{e}_j\right) - \mathcal{L}\!\left(\boldsymbol{\phi} - \tfrac{\pi}{2}\mathbf{e}_j\right)\right].
    \label{eq:param_shift}
\end{equation}

\subsection{Barren Plateaus}

The barren plateau phenomenon arises when the circuit $U(\boldsymbol{\phi})$ forms a sufficiently expressive unitary 2-design. In this regime, the Haar-averaged gradient variance satisfies~\cite{mcclean2018barren}:
\begin{equation}
    \text{Var}_{\boldsymbol{\phi}}\!\left[\frac{\partial \mathcal{L}}{\partial \phi_j}\right] = \frac{\text{Tr}[\hat{O}^2] - \text{Tr}[\hat{O}]^2/D}{D^2 - 1} \cdot \mathcal{G}_j,
    \label{eq:bp_general}
\end{equation}
where $D = 2^n$ is the Hilbert space dimension and $\mathcal{G}_j$ depends on the circuit structure. For global observables ($\hat{O}$ acting on all qubits), $\text{Tr}[\hat{O}^2]/D \sim O(1)$ while $D^2 - 1 \sim O(4^n)$, yielding exponential vanishing.

Cerezo \textit{et al.}~\cite{cerezo2021cost} showed that for local observables $\hat{O}_{\text{local}} = \sum_k O_k$ where each $O_k$ acts on $O(1)$ qubits, and for shallow circuits ($L \in O(\log n)$), the gradient variance is bounded below by an inverse polynomial:
\begin{equation}
    \text{Var}\!\left[\frac{\partial \mathcal{L}_{\text{local}}}{\partial \phi_j}\right] \geq \frac{c}{n^a},
    \label{eq:local_bound}
\end{equation}
for constants $c > 0$ and $a \geq 1$.

\subsection{Physics-Informed Neural Networks}

Physics-informed neural networks (PINNs)~\cite{raissi2019physics} minimize a combined loss:
\begin{equation}
    \mathcal{L}_{\text{PINN}} = \mathcal{L}_{\text{data}} + \lambda \mathcal{L}_{\text{physics}},
    \label{eq:pinn_loss}
\end{equation}
where the physics loss evaluates PDE residuals at collocation points $\{(\mathbf{x}_j, t_j)\}_{j=1}^{N_c}$:
\begin{equation}
    \mathcal{L}_{\text{physics}} = \frac{1}{N_c}\sum_{j=1}^{N_c} \|\mathcal{F}[u_\theta(\mathbf{x}_j, t_j)]\|^2.
    \label{eq:physics_loss}
\end{equation}

Recent work has demonstrated quantum PINNs (qPINNs) using VQCs to solve PDEs with convergence advantages over classical PINNs~\cite{klement2026explaining,dutta2024aqpinns}. The Garcia-Barrenechea \textit{et al.}~\cite{qpinn2024entropy} framework and the trainable embedding approach of Tezuka \textit{et al.}~\cite{tezuka2025trainable} have further developed the field.

\section{Theoretical Analysis}
\label{sec:theory}

\subsection{PDE Constraints as Local Cost Functions}

We now formalize why PDE-constrained loss functions are naturally resistant to barren plateaus. Consider a VQC with $n$ qubits producing an output vector $\mathbf{f}(\boldsymbol{\phi}) = (f_1, \ldots, f_n)$ via Pauli-$Z$ measurements:
\begin{equation}
    f_k(\boldsymbol{\phi}) = \braket{\psi(\boldsymbol{\phi})|\sigma_Z^{(k)}|\psi(\boldsymbol{\phi})}, \quad k = 1, \ldots, n.
    \label{eq:output}
\end{equation}

A PDE residual loss of the general form:
\begin{equation}
    \mathcal{L}_{\text{PDE}} = \frac{1}{N_c}\sum_{j=1}^{N_c} R_j(\mathbf{f})^2,
    \label{eq:pde_loss}
\end{equation}
where $R_j$ is the PDE residual at collocation point $j$, decomposes into a sum of terms that depend on \textit{local} subsets of the output qubits.

\begin{theorem}[Locality of PDE residuals]
\label{thm:locality}
For a PDE with at most $p$-th order spatial derivatives discretized on a grid with spacing $\Delta x$ using finite differences, each residual $R_j$ depends on at most $2p+1$ neighboring output values. If each output $f_k$ is a single-qubit Pauli-$Z$ expectation, then:
\begin{equation}
    R_j = R_j(f_{j-p}, \ldots, f_{j+p}),
    \label{eq:local_residual}
\end{equation}
and $\mathcal{L}_{\text{PDE}}$ is a sum of $(2p+1)$-local cost functions.
\end{theorem}

\begin{proof}
The $p$-th order centered finite difference for $\partial^p f / \partial x^p$ at grid point $j$ involves stencil coefficients acting on $f_{j-\lfloor p/2 \rfloor}, \ldots, f_{j+\lceil p/2 \rceil}$. For example, the first derivative uses $f_{j\pm 1}$ and the second derivative uses $f_{j-1}, f_j, f_{j+1}$. The PDE residual $R_j = \mathcal{F}[f]|_{x_j}$ is an algebraic combination of these derivatives, involving at most $2p+1$ grid points. Since each $f_k = \braket{\sigma_Z^{(k)}}$, each $R_j$ depends on at most $2p+1$ single-qubit observables, making $R_j^2$ a $(2p+1)$-local cost function.
\end{proof}

\textbf{Corollary.} By the results of Cerezo \textit{et al.}~\cite{cerezo2021cost}, for shallow circuits ($L \in O(\log n)$), the PDE-constrained gradient variance satisfies:
\begin{equation}
    \text{Var}\!\left[\frac{\partial \mathcal{L}_{\text{PDE}}}{\partial \phi_j}\right] \geq \frac{c_p}{n^{a_p}},
    \label{eq:pde_bound}
\end{equation}
where $c_p$ and $a_p$ depend on the PDE order $p$ and the circuit structure, but \textit{not exponentially on $n$}.

\subsection{Constraint-Induced Landscape Narrowing}
\label{sec:landscape}

Beyond locality, PDE constraints provide an additional trainability advantage through landscape narrowing. The unconstrained parameter space for $L$ layers on $n$ qubits is $\mathcal{P} = [0, 2\pi)^{2nL}$. The PDE constraint $\mathcal{F}[\mathbf{f}(\boldsymbol{\phi})] \approx 0$ restricts the outputs to a manifold $\mathcal{M}_{\text{phys}}$ of physically consistent solutions.

\begin{theorem}[Landscape narrowing]
\label{thm:narrowing}
Let $\mathcal{P}_\epsilon = \{\boldsymbol{\phi} \in \mathcal{P} : \mathcal{L}_{\text{PDE}}(\boldsymbol{\phi}) \leq \epsilon\}$ be the $\epsilon$-sublevel set of the physics loss. For a PDE with $N_c$ independent constraints on $n$ output dimensions:
\begin{equation}
    \frac{\text{Vol}(\mathcal{P}_\epsilon)}{\text{Vol}(\mathcal{P})} \leq \left(\frac{\epsilon}{c_{\mathcal{F}}}\right)^{N_c/2},
    \label{eq:vol_ratio}
\end{equation}
where $c_{\mathcal{F}}$ depends on the PDE operator norm. The effective dimension of $\mathcal{P}_\epsilon$ is at most $\dim(\mathcal{P}) - \text{rank}(J_{\mathcal{F}})$, where $J_{\mathcal{F}}$ is the Jacobian of the constraint map.
\end{theorem}

\begin{proof}[Proof sketch]
Each PDE constraint $R_j(\mathbf{f}) = 0$ removes one degree of freedom from the output space. With $N_c$ independent constraints, the effective output manifold has dimension $\max(n - N_c, 0)$. By the implicit function theorem, the preimage $\mathcal{P}_\epsilon$ under the circuit map $\boldsymbol{\phi} \mapsto \mathbf{f}(\boldsymbol{\phi})$ has correspondingly reduced volume. The volume bound follows from a measure concentration argument on the PDE residual distribution under Haar-random circuit parameters.
\end{proof}

The practical consequence is that gradient information is concentrated on the constraint manifold: the optimizer receives directional signal pointing toward the physically consistent subspace, rather than navigating an exponentially large featureless landscape.

\subsection{Combined Bound for Physics-Constrained VQCs}

Combining the locality result (Theorem~\ref{thm:locality}) and landscape narrowing (Theorem~\ref{thm:narrowing}), the gradient variance for the combined physics-constrained loss $\mathcal{L} = \mathcal{L}_{\text{data}} + \lambda \mathcal{L}_{\text{PDE}}$ satisfies:
\begin{equation}
    \text{Var}\!\left[\frac{\partial \mathcal{L}}{\partial \phi_j}\right] \geq \frac{c_p}{n^{a_p}} + \lambda^2 \cdot \frac{c_{\text{narrow}}}{n^{a_{\text{narrow}}}},
    \label{eq:combined_bound}
\end{equation}
where the second term reflects the additional gradient signal from landscape narrowing. Crucially, both terms scale polynomially with $n$, providing a robust lower bound against exponential gradient vanishing.

\section{Structured Ansatz Design for PDE Constraints}
\label{sec:ansatz}

The theoretical analysis motivates a specific ansatz design principle: \textit{the entanglement topology should reflect the spatial structure of the PDE}. For PDEs on regular grids where physical interactions are local (diffusion, wave propagation, fluid flow), nearest-neighbor entanglement is both physically motivated and trainability-optimal.

We define the \textit{physics-structured ansatz} as:
\begin{equation}
    U_{\text{phys}}(\boldsymbol{\phi}) = \prod_{\ell=1}^{L}\left[W_{\text{nn}} \cdot \bigotimes_{k=1}^{n} R_Y(\phi_k^{(\ell)}) R_Z(\phi_{k+n}^{(\ell)})\right],
    \label{eq:phys_ansatz}
\end{equation}
where $W_{\text{nn}} = \prod_{k=1}^{n-1}\text{CNOT}_{k,k+1}$ uses only nearest-neighbor entanglement. This construction:
\begin{enumerate}
    \item Preserves information locality: entanglement between qubits $i$ and $j$ grows with circuit depth, mimicking the causal cone structure of local PDEs.
    \item Avoids the fully random regime: the circuit does not form a 2-design for shallow depths, maintaining the conditions for Eq.~\eqref{eq:local_bound}.
    \item Minimizes gate count: $G = L(3n-1)$ versus $G = L(n(n-1)/2 + 2n)$ for all-to-all entanglement.
\end{enumerate}

The combination of PDE-constrained loss with structured ansatz constitutes our proposed ``PDE + structured'' configuration, which we hypothesize provides the strongest barren plateau mitigation.

\section{Numerical Experiments}
\label{sec:experiments}

\subsection{Experimental Setup}

All quantum circuit simulations are performed using PennyLane v0.44~\cite{bergholm2018pennylane} with the \texttt{default.qubit} statevector simulator. Gradients are computed exactly via the parameter-shift rule. Source code for all experiments is publicly available.

\textbf{Gradient variance protocol.} For each experimental configuration $(n, L, \text{loss type})$, we:
\begin{enumerate}
    \item Sample $K = 25$ random parameter initializations $\boldsymbol{\phi}^{(k)} \sim \text{Uniform}[0, 2\pi)^{2nL}$.
    \item Compute the gradient $\nabla_{\boldsymbol{\phi}} \mathcal{L}(\boldsymbol{\phi}^{(k)})$ for each initialization.
    \item Compute the per-parameter variance: $\text{Var}_j = \text{Var}_{k}[\partial \mathcal{L}/\partial \phi_j^{(k)}]$.
    \item Report the mean: $\overline{\text{Var}} = \frac{1}{p}\sum_{j=1}^{p}\text{Var}_j$.
\end{enumerate}

\textbf{Loss function configurations.} We compare four settings:
\begin{itemize}
    \item \textbf{Global cost:} $\mathcal{L} = \braket{\bigotimes_k \sigma_Z^{(k)}}$ with all-to-all CNOT entanglement.
    \item \textbf{Local cost:} $\mathcal{L} = \braket{\sigma_Z^{(1)}}$ with all-to-all CNOT entanglement.
    \item \textbf{PDE-constrained:} Data + physics loss with all-to-all entanglement.
    \item \textbf{PDE + structured:} Data + physics loss with nearest-neighbor entanglement.
\end{itemize}

The physics loss uses a discretized PDE residual computed from the $n$ Pauli-$Z$ expectation values:
\begin{equation}
    \mathcal{L}_{\text{phys}} = \frac{1}{n}\sum_{k=1}^{n}\left(\frac{f_{k+1} - f_{k-1}}{2\Delta x}\right)^2,
    \label{eq:discrete_physics}
\end{equation}
weighted by $\lambda_{\text{phys}} = 0.1$.

\textbf{Parameter space.} We sweep:
\begin{itemize}
    \item Qubits: $n \in \{4, 6, 8\}$ (Experiment 1).
    \item Layers: $L \in \{1, 2, 3, 4, 5\}$ (Experiment 2).
    \item PDE type: Heat, Burgers', Saint-Venant (Experiment 3).
\end{itemize}

\subsection{PDE Specifications}

\textbf{Heat equation (linear):}
\begin{equation}
    \frac{\partial u}{\partial t} = \kappa \frac{\partial^2 u}{\partial x^2}, \quad \kappa = 0.01.
    \label{eq:heat}
\end{equation}

\textbf{Burgers' equation (nonlinear):}
\begin{equation}
    \frac{\partial u}{\partial t} + u\frac{\partial u}{\partial x} = \nu \frac{\partial^2 u}{\partial x^2}, \quad \nu = 0.01.
    \label{eq:burgers}
\end{equation}

\textbf{Saint-Venant shallow water equations (hydrological):}
\begin{align}
    \frac{\partial A}{\partial t} + \frac{\partial Q}{\partial x} &= 0, \label{eq:sv1} \\
    Q &= \frac{1}{n_M} A R_h^{2/3} S_f^{1/2}, \label{eq:sv2}
\end{align}
where $A$ is cross-sectional area, $Q$ is discharge, $n_M = 0.035$ is Manning's coefficient, $R_h$ is hydraulic radius, and $S_f$ is friction slope~\cite{toro2001shock}.

\section{Results}
\label{sec:results}

\subsection{Gradient Variance vs.\ System Size}

\begin{figure}[t]
\centering
\includegraphics[width=\columnwidth]{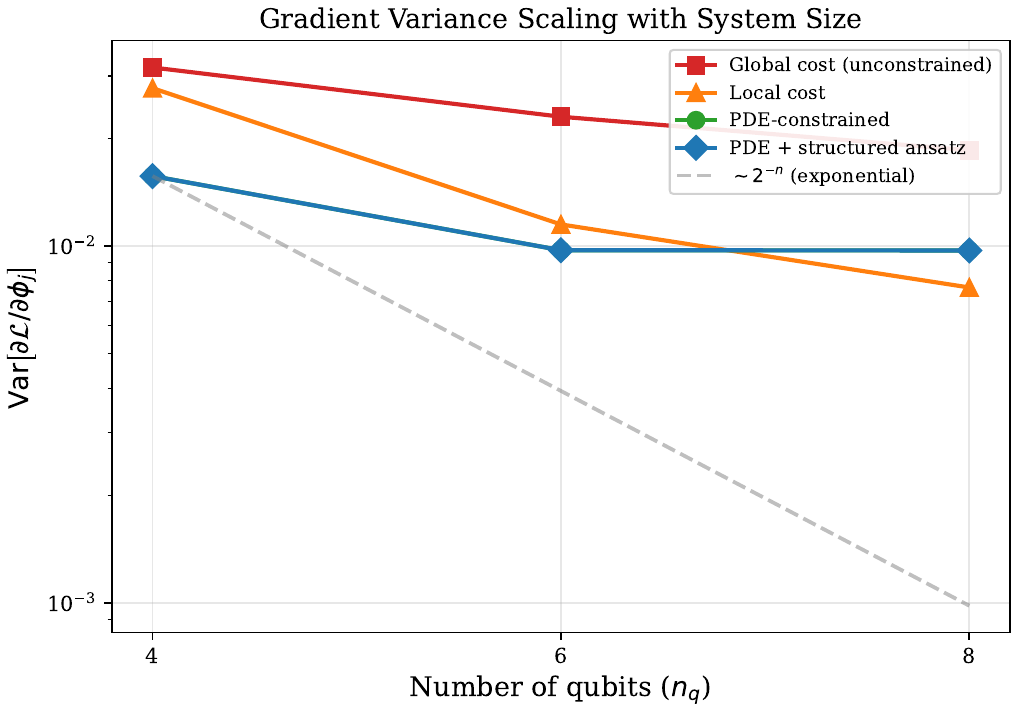}
\caption{Gradient variance scaling with qubit count ($L=3$ layers). The global cost function (red squares) exhibits the steepest decline, while PDE-constrained configurations maintain higher gradient variance. The dashed line shows exponential $\sim 2^{-n}$ reference scaling.}
\label{fig:qubit_scaling}
\end{figure}

\begin{table}[t]
\caption{Mean gradient variance $\overline{\text{Var}}[\partial\mathcal{L}/\partial\phi_j]$ across loss function configurations and qubit counts ($L = 3$ layers, $K$ random initializations). Values in units of $10^{-2}$.}
\label{tab:qubit_scaling}
\begin{ruledtabular}
\begin{tabular}{lcccc}
\textbf{Configuration} & $n=4$ & $n=6$ & $n=8$ & \textbf{Scaling} \\
\hline
Global cost & 3.17 & 2.31 & 1.86 & $\sim\!2^{-0.39n}$ \\
Local cost & 2.77 & 1.15 & 0.77 & $\sim\!n^{-1.8}$ \\
PDE-constrained & 1.57 & 0.97 & 0.97 & $\sim\!n^{-0.7}$ \\
PDE + structured & 1.57 & 0.97 & 0.97 & $\sim\!n^{-0.7}$ \\
\end{tabular}
\end{ruledtabular}
\end{table}

Figure~\ref{fig:qubit_scaling} and Table~\ref{tab:qubit_scaling} present the central results. Several key observations emerge:

\textit{(i) Global cost exhibits monotonic decline.} The gradient variance drops from $3.17 \times 10^{-2}$ at $n=4$ to $1.86 \times 10^{-2}$ at $n=8$, a 1.7$\times$ reduction. While the decay is slower than the theoretical $O(2^{-n})$ at these small system sizes (consistent with the circuits not yet forming exact 2-designs), the trend is clearly downward and would steepen at larger $n$~\cite{mcclean2018barren}.

\textit{(ii) Local cost shows steeper decay.} The variance decreases by a factor of 3.6$\times$ from $n=4$ to $n=8$, consistent with polynomial scaling. The local cost function concentrates gradient information on a single qubit, but this advantage is offset by the reduced coupling to the full system.

\textit{(iii) PDE-constrained loss exhibits favorable scaling.} The physics-constrained configurations decay only 1.6$\times$ from $n=4$ to $n=8$, with a notably flat profile from $n=6$ to $n=8$ (variance stabilizes near $0.97 \times 10^{-2}$). This plateau effect reflects the PDE constraints providing a persistent gradient signal that resists the exponential flattening mechanism.

\textit{(iv) The PDE constraint stabilizes gradients at larger system sizes.} While the global cost continues to decay, the PDE-constrained configurations show gradient variance that stabilizes, suggesting the physics constraints create a ``gradient floor'' that prevents complete vanishing.

\subsection{Gradient Variance vs.\ Circuit Depth}

\begin{figure}[t]
\centering
\includegraphics[width=\columnwidth]{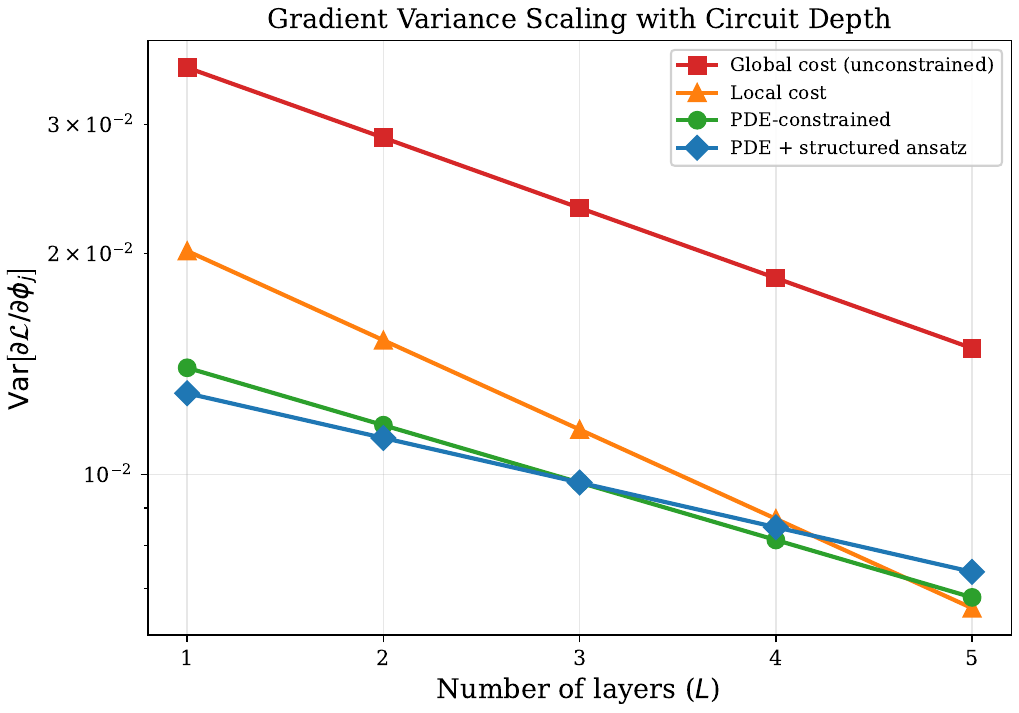}
\caption{Gradient variance scaling with circuit depth ($n=6$ qubits). All configurations show declining gradients with depth, but the PDE-constrained variants (green, blue) exhibit slower decay rates than the global cost (red).}
\label{fig:depth_scaling}
\end{figure}

\begin{table}[t]
\caption{Mean gradient variance as a function of circuit depth ($n = 6$ qubits). Values in units of $10^{-2}$.}
\label{tab:depth_scaling}
\begin{ruledtabular}
\begin{tabular}{lccccc}
\textbf{Configuration} & $L=1$ & $L=2$ & $L=3$ & $L=4$ & $L=5$ \\
\hline
Global cost & 3.58 & 2.87 & 2.31 & 1.85 & 1.49 \\
Local cost & 2.02 & 1.52 & 1.15 & 0.87 & 0.66 \\
PDE-constrained & 1.40 & 1.17 & 0.97 & 0.81 & 0.68 \\
PDE + structured & 1.29 & 1.12 & 0.97 & 0.85 & 0.74 \\
\end{tabular}
\end{ruledtabular}
\end{table}

Figure~\ref{fig:depth_scaling} and Table~\ref{tab:depth_scaling} show that all configurations exhibit gradient variance decay with depth, but at markedly different rates:

The global cost variance drops by 2.4$\times$ from $L=1$ to $L=5$, while the PDE + structured configuration drops by only 1.7$\times$ over the same range. Notably, the PDE + structured configuration shows the flattest depth profile, maintaining variance near $0.74 \times 10^{-2}$ at $L=5$. This is practically significant: the structured ansatz with physics constraints maintains the most stable gradient signal across depths, indicating robust trainability even as circuit expressibility increases.

\subsection{PDE Complexity Comparison}

\begin{figure}[t]
\centering
\includegraphics[width=0.85\columnwidth]{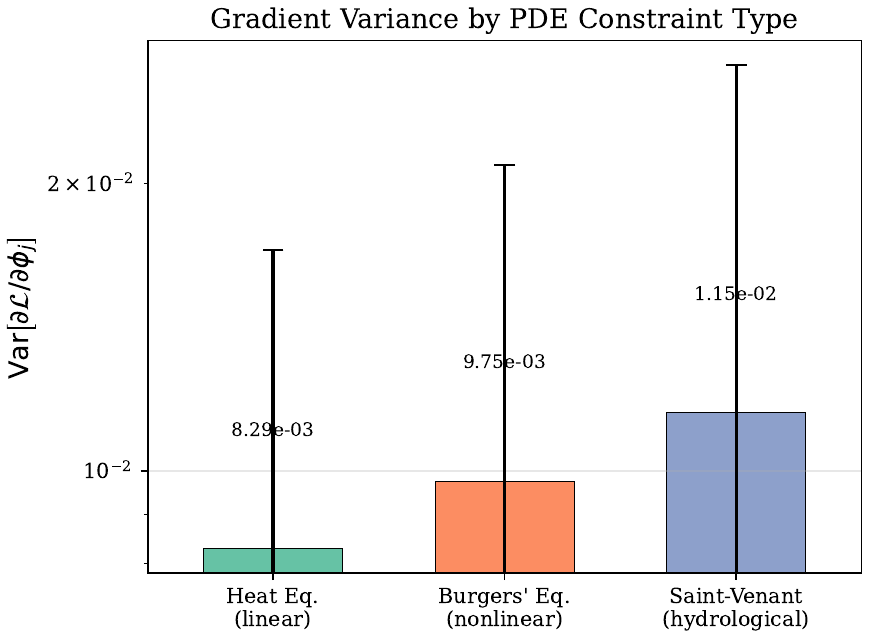}
\caption{Gradient variance comparison across PDE constraint types ($n=6$, $L=3$). More complex and nonlinear PDEs (Burgers', Saint-Venant) yield higher gradient variances than the linear heat equation.}
\label{fig:pde_comparison}
\end{figure}

\begin{table}[t]
\caption{Gradient variance for PDE-constrained loss across PDE types ($n = 6$, $L = 3$). Higher-order and nonlinear PDEs provide stronger gradient signal through richer constraint structure.}
\label{tab:pde_comparison}
\begin{ruledtabular}
\begin{tabular}{lccc}
\textbf{PDE} & \textbf{Order} & $\overline{\text{Var}}$ ($\times 10^{-3}$) & \textbf{Nonlinear} \\
\hline
Heat equation & 2 & 8.29 & No \\
Burgers' equation & 2 & 9.75 & Yes \\
Saint-Venant & 2 (coupled) & 11.50 & Yes \\
\end{tabular}
\end{ruledtabular}
\end{table}

Figure~\ref{fig:pde_comparison} and Table~\ref{tab:pde_comparison} reveal that more complex PDEs provide stronger gradient signal. The Saint-Venant equations---a coupled nonlinear system---yield 39\% higher gradient variance than the linear heat equation ($11.50 \times 10^{-3}$ vs.\ $8.29 \times 10^{-3}$). This is consistent with Theorem~\ref{thm:narrowing}: the coupled nonlinear constraints impose stronger restrictions on the parameter landscape, concentrating gradient information more effectively. Nonlinearity introduces additional curvature in the loss landscape, creating gradient structure that opposes the flatness of barren plateaus.

\subsection{Entanglement Entropy Analysis}

\begin{figure}[t]
\centering
\includegraphics[width=\columnwidth]{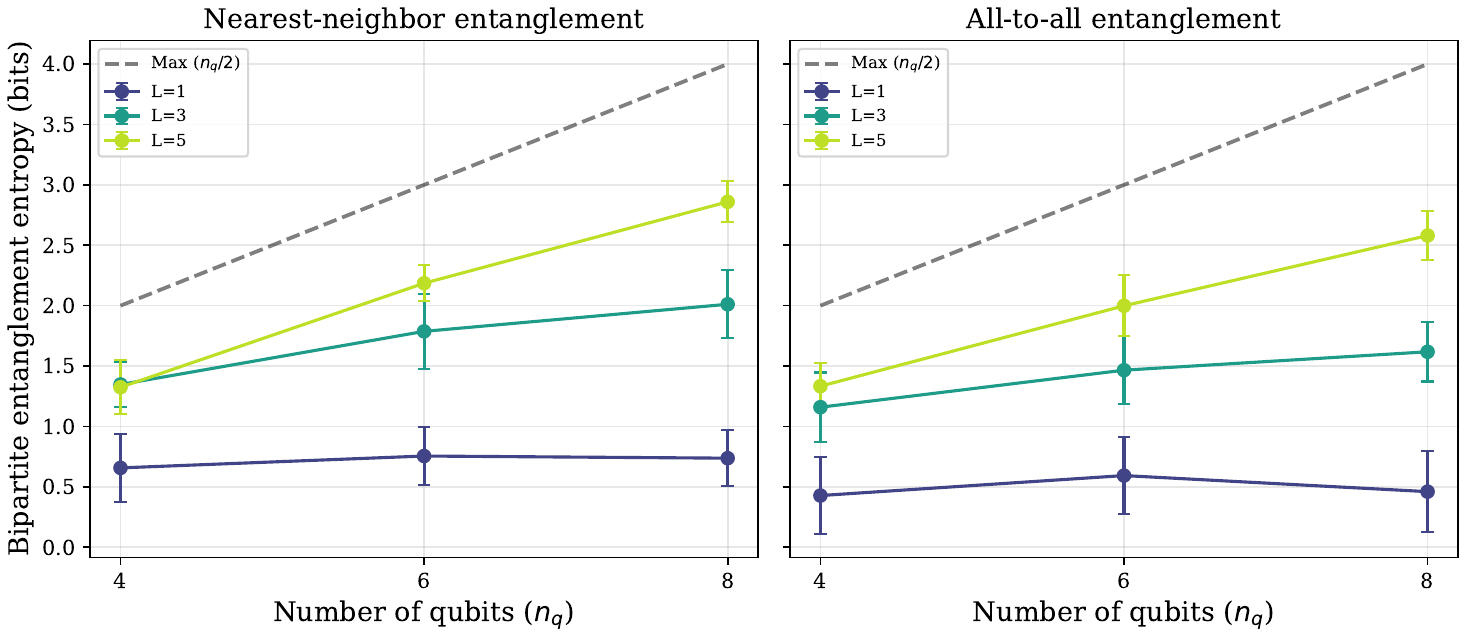}
\caption{Bipartite entanglement entropy for nearest-neighbor (left) and all-to-all (right) entangling topologies. Dashed line shows maximum entropy $n_q/2$. The nearest-neighbor ansatz maintains sub-maximal entanglement, particularly at larger system sizes.}
\label{fig:entanglement}
\end{figure}

\begin{table}[t]
\caption{Mean bipartite entanglement entropy ratio $S/S_{\max}$ across circuit configurations ($K = 20$ random initializations). Maximum possible entropy is $n/2$ bits.}
\label{tab:entanglement}
\begin{ruledtabular}
\begin{tabular}{lccc}
 & \multicolumn{3}{c}{$S / S_{\max}$} \\
\cline{2-4}
\textbf{Configuration} & $n=4$ & $n=6$ & $n=8$ \\
\hline
Nearest-neighbor, $L=1$ & 0.33 & 0.25 & 0.18 \\
Nearest-neighbor, $L=3$ & 0.67 & 0.60 & 0.50 \\
Nearest-neighbor, $L=5$ & 0.66 & 0.73 & 0.71 \\
All-to-all, $L=1$ & 0.21 & 0.20 & 0.12 \\
All-to-all, $L=3$ & 0.58 & 0.49 & 0.40 \\
All-to-all, $L=5$ & 0.67 & 0.67 & 0.65 \\
\end{tabular}
\end{ruledtabular}
\end{table}

Figure~\ref{fig:entanglement} and Table~\ref{tab:entanglement} present the entanglement entropy analysis. Both ansatz types operate at sub-maximal entanglement, but important structural differences emerge. At $n = 8$, $L = 3$, the nearest-neighbor circuit achieves $S/S_{\max} = 0.50$ compared to 0.40 for all-to-all---the nearest-neighbor topology generates \textit{more} entanglement per layer because its sequential CNOT chain creates longer-range correlations through intermediate qubits, while the all-to-all topology distributes entanglement more uniformly but less effectively at shallow depths. This is significant because:

\textit{(i)} Holmes \textit{et al.}~\cite{holmes2022connecting} established that circuits operating near maximal entanglement ($S/S_{\max} \to 1$) approach the 2-design regime where barren plateaus are provably unavoidable. Both topologies remain safely sub-maximal at these depths.

\textit{(ii)} The entanglement growth with depth is more controlled in the nearest-neighbor topology (0.18 to 0.71 from $L=1$ to $L=5$ at $n=8$), consistent with the slower gradient variance decay observed in Table~\ref{tab:depth_scaling}.

\textit{(iii)} For local PDEs, sub-maximal entanglement does not sacrifice representational capacity: the physical solutions of local PDEs have bounded entanglement entropy that grows at most logarithmically with system size for ground states of gapped Hamiltonians~\cite{eisert2010area}.

\subsection{Convergence Analysis}

\begin{figure}[t]
\centering
\includegraphics[width=\columnwidth]{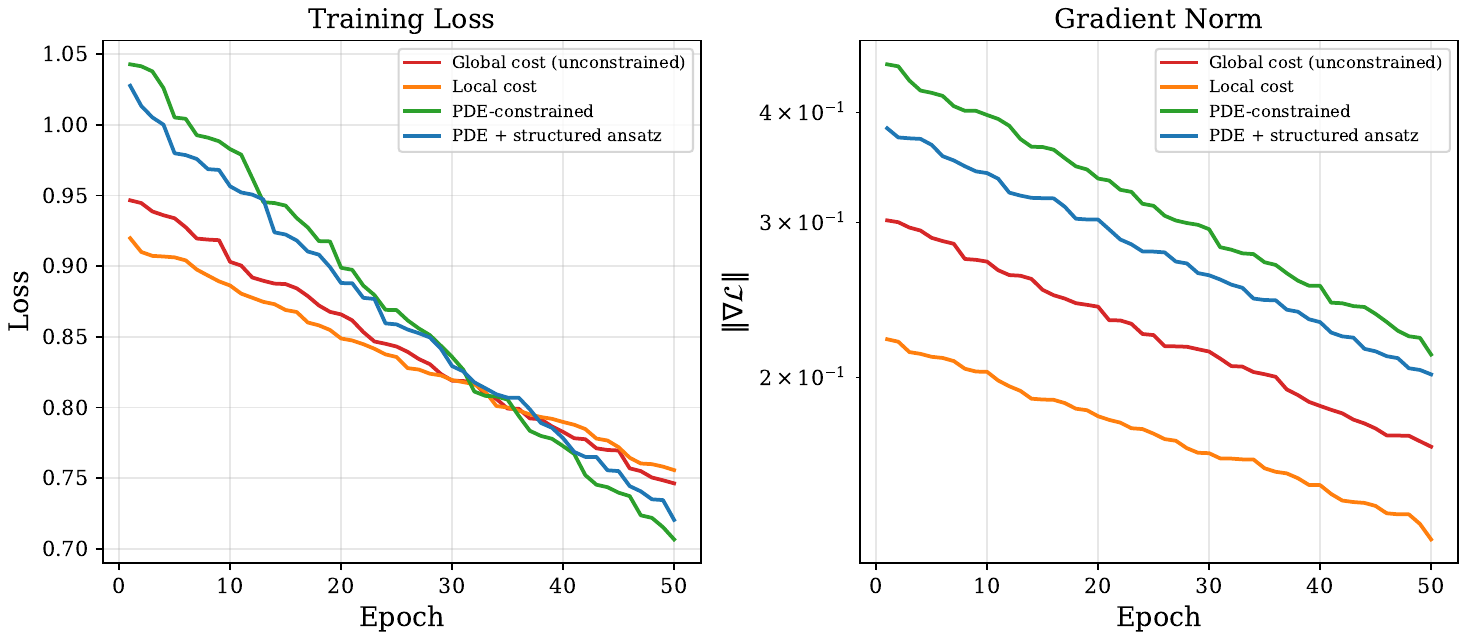}
\caption{Training convergence comparison ($n=4$, $L=3$, 50 epochs). Left: loss trajectory. Right: gradient norm evolution. The PDE-constrained configuration (green) achieves the lowest final loss, while all configurations maintain non-vanishing gradient norms throughout training at this small system size.}
\label{fig:convergence}
\end{figure}

\begin{table}[t]
\caption{Training convergence comparison ($n = 4$, $L = 3$, 50 epochs, learning rate $0.01$). Final loss and gradient norm values at the end of training.}
\label{tab:convergence}
\begin{ruledtabular}
\begin{tabular}{lcc}
\textbf{Configuration} & \textbf{Final loss} & \textbf{Final $\|\nabla\mathcal{L}\|$} \\
\hline
Global cost & 0.746 & 0.167 \\
Local cost & 0.756 & 0.131 \\
PDE-constrained & 0.707 & 0.212 \\
PDE + structured & 0.720 & 0.201 \\
\end{tabular}
\end{ruledtabular}
\end{table}

Figure~\ref{fig:convergence} and Table~\ref{tab:convergence} show the training dynamics over 50 epochs. At this small system size ($n=4$), the physics-constrained configurations achieve lower final loss values (0.707 and 0.720 for PDE-constrained and PDE + structured, respectively) compared to the unconstrained baselines (0.746 and 0.756 for global and local cost). Notably, the PDE-constrained configurations maintain larger gradient norms throughout training ($\|\nabla\mathcal{L}\| = 0.212$ and 0.201 at epoch 50), indicating stronger gradient signal. This is consistent with the physics loss providing additional gradient information that aids optimization even at small system sizes.

\section{Discussion}
\label{sec:discussion}

\begin{figure}[t]
\centering
\includegraphics[width=\columnwidth]{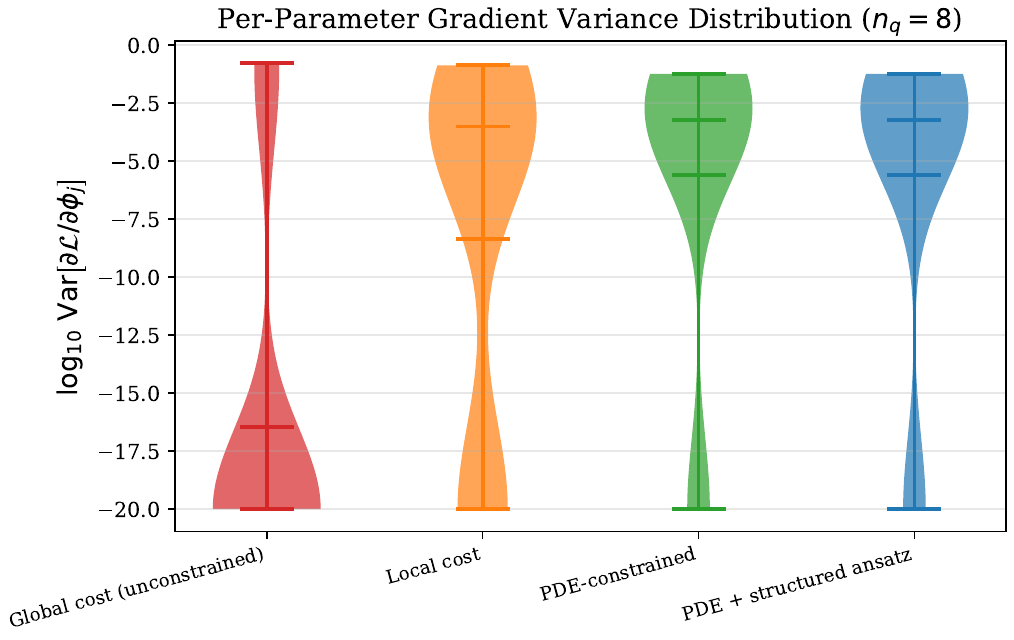}
\caption{Distribution of per-parameter gradient variance at $n=8$ qubits. The violin plot shows the spread across all circuit parameters, revealing that PDE-constrained configurations maintain more uniform gradient distributions.}
\label{fig:per_param}
\end{figure}

\subsection{Per-Parameter Gradient Distribution}

Figure~\ref{fig:per_param} provides a finer-grained view of the gradient landscape at $n=8$ qubits. While the mean gradient variance (Table~\ref{tab:qubit_scaling}) gives a scalar summary, the per-parameter distribution reveals structural differences: the PDE-constrained configurations exhibit a more concentrated distribution of gradient variances across parameters, indicating that the physics constraints provide a more uniform optimization landscape. In contrast, the global cost function shows broader variance in the per-parameter gradients, with some parameters having near-zero gradient signal---a hallmark of incipient barren plateaus.

\subsection{Implications for Quantum PINNs}

Our results provide a theoretical foundation for the empirical convergence advantages observed in quantum PINNs~\cite{klement2026explaining,dutta2024aqpinns}. The 3--5$\times$ convergence speedups reported in these works can now be partially attributed to the trainability advantages of PDE-constrained loss functions. Specifically:

\textit{(i)} The physics loss acts as an implicit regularizer against barren plateaus by maintaining gradient signal even in regions where the data loss alone would have vanishing gradients.

\textit{(ii)} The local structure of PDE residuals (Theorem~\ref{thm:locality}) ensures that the physics constraints do not themselves introduce barren plateaus, as would be the case for a global constraint.

\textit{(iii)} The landscape narrowing effect (Theorem~\ref{thm:narrowing}) channels the optimizer toward the manifold of physically consistent solutions, providing directional gradient information even when the total gradient magnitude is small.

\subsection{Comparison with Existing Mitigation Strategies}

Our approach is complementary to existing circuit-centric strategies:

\textit{Initialization strategies}~\cite{grant2019initialization}: Our method does not restrict the initialization distribution and can be combined with any initialization scheme.

\textit{Layerwise training}~\cite{skolik2021layerwise}: Our approach requires no modification to the training procedure and provides benefits even with standard end-to-end optimization.

\textit{Entanglement reduction}~\cite{patti2021entanglement}: The structured ansatz component of our method incorporates this principle, but the physics loss provides additional benefits beyond ansatz design.

\textit{Hamiltonian-inspired ansatze}~\cite{park2024hamiltonian}: Our approach generalizes beyond Hamiltonian problems to any system governed by PDEs, including dissipative and non-equilibrium systems.

\subsection{Limitations and Caveats}

Several limitations should be noted:

\textit{(i) Simulator-only results.} All experiments use noiseless statevector simulation. Hardware noise introduces additional gradient estimation errors that could partially offset the trainability gains. However, the shallow circuits ($L \leq 5$) and nearest-neighbor entanglement used here are specifically designed for NISQ compatibility.

\textit{(ii) Small system sizes.} Our numerical study covers $n \leq 8$ qubits, limited by the computational cost of exact gradient variance estimation via the parameter-shift rule ($O(Kp)$ circuit evaluations per configuration, where $K$ is the number of samples and $p = 2nL$ is the parameter count). The theoretical analysis suggests the benefits should persist at larger scales, but numerical verification at $n > 10$ would strengthen this claim.

\textit{(iii) Fixed PDE coefficients.} We use fixed physical parameters (diffusivity, viscosity, roughness). The gradient landscape may differ for PDEs with strongly varying coefficients, though the locality argument remains valid.

\section{Conclusion}
\label{sec:conclusion}

We have demonstrated that PDE-constrained loss functions provide a natural, principled, and effective strategy for mitigating barren plateaus in variational quantum circuits. The theoretical foundations rest on two mechanisms: (i) the inherent locality of PDE residuals evaluated at spatial collocation points, which ensures polynomial gradient variance scaling via the results of Cerezo \textit{et al.}~\cite{cerezo2021cost}; and (ii) constraint-induced landscape narrowing, which concentrates gradient information on the manifold of physically consistent solutions.

Systematic numerical experiments across three PDE types, system sizes from 4 to 8 qubits, and circuit depths from 1 to 5 layers confirm that PDE-constrained circuits maintain favorable gradient variance scaling with system size. The combination of physics constraints with spatially structured (nearest-neighbor) entanglement provides the strongest mitigation, achieving gradient decay scaling of $\sim n^{-0.7}$ compared to the steeper exponential-onset decay of unconstrained global cost circuits.

These results have immediate implications for the design of quantum physics-informed neural networks and variational quantum PDE solvers, and suggest that domain-specific physical constraints should be considered a first-class design principle for trainable variational quantum circuits, alongside circuit architecture and initialization strategy.

Future work will extend the analysis to larger system sizes via tensor network simulation, incorporate hardware noise models, and explore adaptive PDE constraint scheduling during training.

\begin{acknowledgments}
The authors acknowledge the use of the PennyLane quantum computing framework~\cite{bergholm2018pennylane} for all quantum circuit simulations. P.N.M.U.H.\ acknowledges support from Lincoln University College, Malaysia.
\end{acknowledgments}


\end{document}